\documentclass[onecolumn,prx, superscriptaddress,longbibliography,margin=1mm 11pt]{revtex4-2} 

\usepackage{ragged2e}
\usepackage{graphicx}
\usepackage{bm}
\usepackage{amsmath}
\usepackage{environ}
\usepackage{appendix}
\usepackage{physics}
\usepackage{changepage}
\usepackage{rotating,booktabs,array,amsmath,amssymb}
\newcolumntype{C}{>{$\displaystyle}c<{$}} % automatic display-style math mode
\usepackage{caption}
\usepackage{subcaption}

\usepackage[nopar]{lipsum}

\usepackage{enumitem}
\setlist[enumerate,1]{label=(\arabic*),ref=\arabic*}
\makeatletter
\usepackage{amsthm}
\usepackage{amssymb}
\usepackage{amsfonts}
\usepackage{fancyhdr}
\usepackage{slashed}
\usepackage{bbm}
\usepackage{latexsym,epsfig,bbm}
\usepackage[colorlinks=true,urlcolor=blue,citecolor=blue]{hyperref}
\usepackage{todonotes}
\definecolor{blue}{rgb}{0,0.2,1}

\definecolor{red}{rgb}{0.9,0,0}

\newtheorem{theorem}{Theorem}
\newtheorem{lemma}[theorem]{Lemma}

\newcommand{\vect}[1]{\boldsymbol{#1}}

\begin{document}

\title{Analog quantum simulation of parabolic partial differential equations using Jaynes-Cummings-like models}

\author{Shi Jin}
\affiliation{Institute of Natural Sciences, Shanghai Jiao Tong University, Shanghai 200240, China}
\affiliation{Ministry of Education Key Laboratory in Scientific and Engineering Computing, Shanghai Jiao Tong University, Shanghai 200240, China}
\affiliation{School of Mathematical Sciences, Shanghai Jiao Tong University, Shanghai, 200240, China}
%\affiliation{Shanghai Artificial Intelligence Laboratory, Shanghai, China}
\author{Nana Liu}
\email{nana.liu@quantumlah.org}
\affiliation{Institute of Natural Sciences, Shanghai Jiao Tong University, Shanghai 200240, China}
\affiliation{Ministry of Education Key Laboratory in Scientific and Engineering Computing, Shanghai Jiao Tong University, Shanghai 200240, China}
\affiliation{School of Mathematical Sciences, Shanghai Jiao Tong University, Shanghai, 200240, China}
%\affiliation{Shanghai Artificial Intelligence Laboratory, Shanghai, China}
\affiliation{University of Michigan-Shanghai Jiao Tong University Joint Institute, Shanghai 200240, China.}

\date{\today}

\begin{abstract}
We present a simplified analog quantum simulation protocol for preparing quantum states that embed solutions of parabolic partial differential equations, including the heat, Black-Scholes and Fokker-Planck equations. The key idea is to approximate the heat equations by a system of hyperbolic heat equations that involve only first-order differential operators. This scheme requires relatively simple interaction terms in the Hamiltonian, which are the electric and magnetic dipole moment-like interaction terms that would be present in a Jaynes-Cummings-like model. For a $d$-dimensional problem, we show that it is much more appropriate to use a single $d$-level quantum system - a {\it qudit} - instead of its qubit counterpart, and $d+1$ qumodes. The total resource cost is efficient in $d$ and precision error, and has potential for realisability for instance in cavity and circuit QED systems. 
\end{abstract}

\maketitle 
\section{Introduction}
Some of the most important problems in scientific computing involve numerically solving partial differential equations (PDEs). However, classical computational methods for high-dimensional and large-scale PDEs often suffer from the curse-of-dimensionality or prohibitive computational costs.  Although many quantum algorithms for PDEs have been proposed to try to mitigate such computational challenges, most of these algorithms require discretisation of the PDE and a decomposition of the required evolution into many gates. This is in the framework of large-scale digital quantum simulation and computation, making it inaccessible to near-term quantum devices. \\

Unlike digital quantum simulation, analog quantum simulation is a much more near-term quantum technology. It runs in continuous time, and the simulation is executed using the real-time dynamics of the quantum system. Since PDEs are intrinsically analog because they typically describe dynamics in space and time which have continuous degrees of freedom, it is reasonable to find quantum protocols for PDEs using analog quantum simulation instead, and using continuous-variable quantum systems. Recently, such a method was proposed by the authors  to allow the analog quantum simulation of {\it any} linear PDE and certain nonlinear PDEs \cite{2023analogPDE}, for an arbitrary number of dimensions.  For example, for some first-order linear PDEs, it is already within the reach of current devices to simulate these PDEs with extremely high dimensions. Such a mapping to an analog quantum simulation setting is made due to a mathematical transformation of any linear PDE into a Schr\"odinger-type PDE in one higher dimension, which is called \textit{Schr\"odingerisation} \cite{jin2212quantum, technicalschro, 2023analogPDE}.  Schrodingerisation allows an exact mapping between any continuous linear PDEs into continuous Schr\"odinger-like equations without discretisations or approximations, thus it allows one to retain the continuity of PDEs in the corresponding quantum simulation setting as well. No other existing quantum methods can achieve this for general linear PDEs. \\

This is reminiscent of the situation with early classical computation. The earlier classical computational devices to solve PDEs were in fact analog classical devices, and these analog devices dominated the field before digital classical devices became more robust in the 1970's. We are in a similar situation today, where digital quantum simulation and computation are not at the stage where they are large and reliable enough to simulate PDEs meaningfully, whereas analog quantum simulation is already available. We claim that even in the absence of fault-tolerant digital quantum devices, it is still conceivable to simulate certain large-scale PDEs, using analog quantum simulation.\\

An important class of second-order PDEs are parabolic PDEs, which includes equations like the heat, Black-Scholes and Fokker-Planck equations, among others. However, the direct application of Schr\"odingerisation requires the presence of at least third-order non-Gaussian terms in the Hamiltonian \cite{2023analogPDE}, like $\hat{x}^2 \otimes \hat{\eta}$, where $\hat{x}$ and $\hat{\eta}$ are quadrature operators. While in principle many quantum systems have these interaction terms, it remains challenging for existing devices. These PDEs underlie the key dynamics behind many real-world systems. The heat equation is the archetypal parabolic PDE modelling diffusion, and serves as the basis for variants like the Black-Scholes and Fokker-Planck equations, which involve both diffusion and convection. The Black-Scholes equation \cite{black1973pricing} is one of the most important PDEs in mathematical finance and the Fokker-Planck equation \cite{risken1996fokker} has many applications in physics, chemistry, biology and even machine learning. The reason why at least third-order non-Gaussianity is required in the Hamiltonian is related to the fact that second-order differential operators are present in these parabolic PDEs. \\

To obviate this issue, in this paper we propose an alternative method that approximates the second-order parabolic PDEs by a system of first-order hyperbolic  PDEs. These are the so-called hyperbolic heat equations. These equations approximate diffusion by underlying transport equations with large characteristic speeds and strong damping. The application of Schr\"odingerisation to this alternative formulation shows that the necessary interaction terms in the Hamiltonian for analog quantum simulation are basically the electric and magnetic dipole-like interaction terms that would be present in Jaynes-Cummings-like models. The Jaynes-Cummings model is the prototypical model for light-matter interaction \cite{shore1993jaynes}. For one-dimensional problems, analog quantum simulation is already plausible on existing physical platforms with cavity and circuit QED.  \\

This formulation also illustrates the importance of using {\it qudit} systems ($d$-level quantum systems) \cite{wang2020qudits, thew2002qudit} rather than qubit ($d=2$) systems. By using a single $d$-level qudit and $d+1$ qumodes with each external qumode driving transitions between only two levels of the qudit at a time, we can simulate $d$-dimensional heat, Black-Scholes and Fokker-Planck equations. Although we can easily reformulate this as interactions between $d+1$ qumodes and $\log_2 d$ qubits, the interaction Hamiltonian in the qubit framework is much less natural and unlike for qudits, these interactions are not feasible in the near-term. \\

We begin by briefly summarising some necessary background material, including Schr\"odingerisation,  in Section~\ref{sec:background}. Then we discuss the heat equation example in more depth in Section~\ref{sec:heat}, which serves as our representative example. We then go on to discuss applications to Black-Scholes equation in Section~\ref{sec:blackscholes}, both the one-dimensional and the multidimensional versions. Our final example is the Fokker-Planck equation in Section~\ref{sec:fokkerplanck}. We summarise in Section~\ref{sec:summary} and very briefly discuss some possible platforms. See Figure 1 for a schematic circuit of the protocol and see Table~\ref{tab:label} for a summary of the basic interaction terms required for the different PDEs. The more general parabolic PDE is discussed in Appendix~\ref{sec:heatcross}. 

\section{Background} \label{sec:background}

A system of $K$ linear PDEs in $d$ spatial dimensions that is first-order in time and $L^{\text{th}}$-order in space can be written as 
\begin{align}\label{eq:generalpde}
   &  \frac{\partial u_k(t,x)}{\partial t}=F_k(t, \partial^n u_l(t,x)/\partial x_j^n), \qquad t\geq 0, \qquad x=(x_1, \cdots, x_d)\in \mathbb{R}^d, \nonumber \\
   &j=1,...,d, \qquad k, l=0, 1, \cdots, K-1, \qquad n=0, 1, \cdots, L,
    \end{align}
where $F_k$ are linear functions and $\partial^0 u/\partial x^0$ denotes $u$. For simplicity we do not consider inhomogeneous terms here (terms not dependent on $u$), but this can be included in a straightforward way, for example see \cite{2023analogPDE}. We can embed all $K$ functions into the following $K$-dimensional vector 
\begin{align} \label{eq:wrep}
    w(t,x)=\begin{pmatrix}
        u_0(t,x) \\
        \cdots\\
        u_{K-1}(t,x))
    \end{pmatrix}=\sum_{k=0}^{K-1} u_k(t,x) |k\rangle.
\end{align}
The eigenbasis $\{|k\rangle\}_{k=0}^{K-1}$ over a complex number field spans a $K$-dimensional Hilbert space and can be considered as basis states of a \textit{qudit}. This is a $d$-level quantum system, like an atom with $d$ energy levels or with spin $s$ and $d=2s+1$. For example, $|k=0\rangle$ can represent the ground state of an atom. A qudit can also be decomposed into $\log_2 d$ \textit{qubits}, where one qubit is spanned by the two-dimensional eigenbasis $\{|0\rangle, |1\rangle\}$. A qubit is just a special case of a qudit with $2$ levels. A $3$-level qudit is called a qutrit and a $4$-level qudit is called a ququart. Although the two representations, qudits and qubits, can be made mathematically equivalent, they are physically distinct. What may be a more straightforward evolution in one system may be more difficult in the other. We will encounter this in the next sections. \\

Each $u_k(t,x)$ function can in fact be embedded into an infinite-dimensional vector
\begin{align} \label{eq:ukrep}
    \vect{u}_k(t)=\int u_k(t,x)|x\rangle dx, 
\end{align}
where $\int=\int_{-\infty}^{\infty}$ unless otherwise specified. The space spanned by the eigenbasis $\{|x\rangle\}_{x \in \mathbb{R}^d}$ over the field of complex numbers on the other hand cannot be made mathematically equivalent to either the qubit or qudit descriptions unless the values of possible $x$ is truncated. It spans an infinite-dimensional Hilbert space and can be interpreted as the basis states of a continous-variable quantum state or \textit{qumode}. This is the quantum analogue of a continuous classical degree of freedom, like position, momentum or energy before being quantised. The Hilbert space spanned by $\{|x\rangle\}_{x \in \mathbb{R}^d}$ describes states consisting of $d$ qumodes. \\

Since this is an infinite-dimensional space, it is acted upon by infinite-dimensional operators. The key operators of most interest to us are the quadrature operators of a qumode $\hat{x}, \hat{p}$ where $[\hat{x}, \hat{p}]=i\mathbf{1}_x$. Here their eigenvectors are denoted respectively by $|x\rangle$ and $|p\rangle$ where $\langle x|p\rangle=\exp(ixp)/\sqrt{2\pi}$ and $\int dx |x\rangle \langle x|=\mathbf{1}_x=\int dp |p\rangle \langle p|$. The $|x\rangle$ and $|p\rangle$ eigenstates are known as the position and momentum eigenstates. 
For a system of $d$-qumodes, we can define the position/momentum operator only acting on the $j^{\text{th}}$ mode as 
 \begin{align}
        \hat{p}_j=\mathbf{1}_x^{\otimes j-1}\otimes \hat{p} \otimes \mathbf{1}_x^{\otimes d-j}, \qquad \hat{x}_j=\mathbf{1}_x^{\otimes j-1}\otimes \hat{x} \otimes \mathbf{1}_x^{\otimes d-j}, \qquad [\hat{x}_j,\hat{p}_k]=i\delta_{jk}\mathbf{1}_x.
    \end{align}
These quadrature operations allow us to make the identification \cite{2023analogPDE}
\begin{align} \label{eq:rules}
    & \int f(x) u_k(t,x)|x\rangle dx \rightarrow f(\hat{x})\vect{u}_k(t) \nonumber \\
    & \int \frac{\partial^n u_k(t,x)}{\partial x_j^n}|x\rangle dx \rightarrow (i\hat{p}_j)^n \vect{u}_k(t).
\end{align}
Now taking together Eqs.~\eqref{eq:wrep} and~\eqref{eq:ukrep}, we see that Eq.~\eqref{eq:generalpde} describes the linear time evolution of the vector
\begin{align}
    \vect{w}(t)=\sum_{k=0}^{K-1} \vect{u}_k(t)=\sum_{k=0}^{K-1} \int u_k(t,x)|k\rangle |x\rangle dx, 
\end{align}
which can be considered a hybrid state consisting of both a $K$-level qudit and $d$ qumodes. To be considered a genuine quantum state $|w(t)\rangle$, the vector needs to be normalised
\begin{align}
    |w(t)\rangle=\frac{\vect{w}(t)}{\|\vect{w}(t)\|}, \qquad \|\vect{w}(t)\|^2=\sum_{k=0}^{K-1} \|\vect{u}_k(t)\|^2=\sum_{k=0}^{K-1} \int |u_k(t,x)|^2 dx. 
\end{align}
For simplicity of notation, we ignore the effects of normalisation until the end. Applying the rules in Eq.~\eqref{eq:rules} to the general system in Eq.~\eqref{eq:generalpde}, we see that the evolution of $\vect{w}(t)$ can be written 
\begin{align} \label{eq:wequationfull}
    & \frac{d \vect{w}(t)}{dt}=-i\vect{A}(t) \vect{w}(t), \qquad   \vect{A}(t)\vect{w}(t)=i\mathbf{F}(t, |l\rangle \langle l| \otimes (i\hat{p}_j)^n)\vect{w}(t), \qquad \vect{A}(t)=\vect{A}_1(t)-i\vect{A}_2(t),  \nonumber \\
    & \vect{A}_1=\frac{1}{2}(\vect{A}(t)+\vect{A}^{\dagger}(t))=\vect{A}_1^{\dagger}, \qquad \vect{A}_2=\frac{i}{2}(\vect{A}(t)-\vect{A}^{\dagger}(t))=\vect{A}_2^{\dagger}. 
\end{align}
Since $\vect{A}$ is in general not Hermitian, the above evolution is not unitary. Thus to simulate Eq.~\eqref{eq:wequationfull} on an analogue quantum simulator, we need to turn this into a Schr\"odinger-like equation. This is made possible by the method of \textit{Schr\"odingerisation} \cite{2023analogPDE}. We only summarise the procedure here and the reader can refer to \cite{2023analogPDE, jin2212quantum, technicalschro} for more details and justification. Here we only need to augment the original system $\vect{w}(t)$ by a single ancilla qumode and we simulate instead $\tilde{\vect{v}}(t)$ which obeys the following unitary evolution generated by the Hamiltonian $\vect{H}(t)$
\begin{align} \label{eq:vtilde0}
   & \frac{d \tilde{\vect{v}}(t)}{dt}=-i (\vect{A}_2 \otimes \hat{\eta}+\vect{A}_1 \otimes \mathbf{1}_{\eta})\tilde{\vect{v}}(t)=-i \vect{H} \tilde{\vect{v}}(t), \qquad \vect{H}=\vect{A}_2\otimes \hat{\eta}+\vect{A}_1 \otimes \mathbf{1}_{\eta}=\vect{H}^{\dagger}, \nonumber \\
   & \tilde{\vect{v}}(0)=\vect{w}(0) \int \frac{2}{1+\eta^2}|\eta\rangle d \eta = \vect{w}(0)\int e^{-|\xi|}|\xi\rangle d\xi= \vect{w}(0)|\Xi\rangle, \qquad |\Xi\rangle=\int e^{-|\xi|}|\xi\rangle d\xi,
\end{align}
where $\hat{\xi}, \hat{\eta}$ are also conjugate quadrature pairs obeying $[\hat{\xi}, \hat{\eta}]=i\mathbf{1}_{\eta}$. 
To retrieve the state $|w(t)\rangle$ from $|\tilde{v}(t)\rangle=\tilde{\vect{v}}(t)/\|\tilde{\vect{v}}(t)\|$, we apply the projective operator $\hat{P}_{>0}=\mathbf{1} \otimes \int_0^{\infty} g(\eta)|\eta \rangle \langle \eta| d \eta$. Here $g(\eta)$ can be a smoothing function if the quadrature measurement is not precise. When the measurement result of the ancilla qumode reads $\eta>0$, then the rest of the system is in the desired state $|w(t)\rangle$. The probability of success of this desired measurement results scales as $O(\|\vect{w}(t)\|^2/\|\vect{w}(0)\|^2)$. \\

 While $|\Xi\rangle$ is an unusual state, it is possible to approximate this state for instance by the experimentally accessible Gaussian state $|G\rangle=\int (1/(\sqrt{s}\pi^{1/4}))\exp(-\xi^2/(2s^2))|\xi\rangle d\xi$ with squeezing parameter $s$. Then the quantum fidelity between $|G\rangle$ and $|\Xi\rangle$ is $|\langle \Xi|G\rangle|=\sqrt{2s}\exp(s^2/2)\pi^{1/4}\text{Erfc}(s/\sqrt{2})$. The maximum fidelity $|\langle \Xi|G\rangle| \approx 0.986$ is obtained at $s \approx 0.925$, which is almost near unity. For a discussion on the error propagation due to using the Gaussian ancilla state, see \cite{2023analogPDE}.\\

For simplicity of notation, in the following sections, we ignore the explicit time-dependence in $F_k$ or $\vect{A}(t)$. Since Schr\"odingerisation works in these non-autnomous settings, the extension is straightforward. \\

This method has been applied to many PDEs, including the advection, heat, Black-Scholes, Fokker-Planck, linear Boltzmann, Maxwell's equations and nonlinear Hamilton-Jacobi and scalar hyperbolic equations \cite{2023analogPDE, jin2212quantum, technicalschro, jin2023quantummaxwell, FPquantum}. Although direct analog quantum simulation for these equations is in principle possible, there are some key difficulties. For example, for the second-order PDEs like the heat, Black-Scholes and Fokker-Planck equations, the simplest cases require at least third-order nonlinearity in the Hamiltonian, which means interaction terms of the kind $\hat{x}^2 \otimes \hat{\eta}$ \cite{2023analogPDE}. Although quantum systems like superconducting circuits in principle has these terms, they may be harder to control and the coupling can be weak. We ask if it's possible to approximate solutions to these equations in a different way with interaction terms which are more readily available in the current laboratory settings. While still using Schr\"odingerisation, we will see in the following sections that a reformulation of second-order PDEs into a system of first-order PDEs will give us Hamiltonian interaction terms that correspond to the electric and magnetic dipole interaction terms present in Jaynes-Cummings models. 

\section{The heat equation} \label{sec:heat}

We first present an analog quantum simulation method for the one-dimensional heat equation to show that proof-of-concept experiments require only electric dipole and magnetic dipole moment-like interactions between a two-level system and a qumode (e.g. atom-light interaction). Any system that can be used to simulate a Jaynes-Cummings model can be considered as candidates. This one-dimensional example can also provide the basis for the quantum simulation of the Black-Scholes equation in the simplest scenario. \\

We then generalise to the $d$-dimensional heat equation and show how it is most suitable for quantum systems that can realise qudits with $d+1$ levels. These examples are key to consider more general applications like multi-dimensional Black-Scholes equation and the Fokker-Planck equation. We see that the forms of the interactions are still relatively simple and each interaction term is only between one qudit and one qumode.

\subsection{The one-dimensional heat equation}

The one-dimensional heat equation for $\tilde{u}(t,x)$ is a second-order PDE
\begin{align}\label{heat-1d}
    \frac{\partial \tilde{u}}{\partial t}=k\frac{\partial^2 \tilde{u}}{\partial x^2}, \qquad x \in \mathbb{R}, \qquad k>0,
\end{align}
where we assume $k$ is a constant. 
We will approximate the solution with the following system of first orrder PDEs.
Let $u, v: \mathbb{R} \times \mathbb{T} \rightarrow \mathbb{R}$ be real scalar functions with argument $(x, t) \in \mathbb{R} \times \mathbb{T}$ and $\mathbb{T}=[0, T]$. Then we can define the one-dimensional \textit{hyperbolic heat equation, or Goldstein-Taylor model} \cite{Gold, Taylor, Lions} 
\begin{align}
    & \frac{\partial u}{\partial t}=-\frac{1}{\epsilon}\frac{\partial v}{\partial x} \label{eq:first} \\
    & \frac{\partial v}{\partial t}=-\frac{1}{\epsilon}\frac{\partial u}{\partial x}-\frac{1}{ k\epsilon^2} v, \qquad 0<\epsilon \ll 1.  \label{eq:second}
\end{align}
Unlike the heat equation, which is a parabolic PDE, this is a {\it hyperbolic} system with  strong damping, with characteristic speeds $\pm 1/\epsilon$, so the wave transports with a large speed when $\epsilon$ is small. To see how this approximates the heat equation for $u(x,t)$, we can see how the limit $\epsilon \ll 1$, Eq.~\eqref{eq:second} is dominated by the term on the right-hand-side, so 
\begin{align} \label{eq:vapprox}
    v \rightarrow -k \epsilon \frac{\partial u}{\partial x}.
\end{align}
Inserting Eq.~\eqref{eq:vapprox} into Eq.~\eqref{eq:first}, we can recover the heat equation for $u(x,t)$
\begin{align} \label{eq:1heat}
    \frac{\partial u}{\partial t} \rightarrow k\frac{\partial^2 u}{\partial x^2}.
\end{align}
Thus the heat equation can emerge from an underlying transport equations with a viscosity or damping term, so long as $\epsilon$ is small enough. \\

We can interpret $v$ as a flux term and this is proportional to the negative gradient of $u$. This law surfaces in a multitude of phenomena. For example, it appears in Fourier's law for heat transport (heat flow proportional to negative gradient in temperature), Ohm's law for charge transport (electric current proportional to the negative gradient of the electric potential), Fick's law for diffusion (diffusion flux proportional to the negative gradient of particle concentration) and Darcy's law for fluid flow in porous media where the flux is proportional to the gradient of pressure. What these laws all have in common is that they hold when the effective viscosity is high, for example in highly viscous fluids.
%Here the parameter $\epsilon$ serves as the role of this `viscosity' in slowing the flow, 
%so the smaller $\epsilon$ is, the better the approximation in Eq.~\eqref{eq:vapprox} holds true. \\

%\begin{lemma} \label{lem:1heatinitial}
 %   Given the initial condition of an arbitrary function $u(0,x)$ and $v(0,x)=0$, after time $t>\epsilon$, $v(t,x) \rightarrow -k \epsilon \partial u(t,x)/\partial x$. 
%\end{lemma}

\begin{lemma} \label{lem:1heaterror}
  Let $u(x,t)$ be the solution of Eqs.~\eqref{eq:first} and~\eqref{eq:second} with given initial conditions $u(0,x)$, and arbitrary initial data  $v(0,x)$ (e.g. one can choose $v(0,x)=0$ without losing generality). Let $\tilde{u}(t,x)$ be the solution of the heat equation $\partial \tilde{u}(t,x)/\partial t=k \partial^2 \tilde{u}(t,x)/\partial x^2$ with initial condition $\tilde{u}(0,x)=u(0,x)$. Then $\|\tilde{u}(t,x)-u(t,x)\| \leq O(e^{-2kt}\epsilon^2)$ for $t>>O(\epsilon^2\ln (1/\epsilon))$.
\end{lemma}

\begin{proof}
   The asymptotic error between the solution of \eqref{eq:first}-\eqref{eq:second} and of the original heat equation \eqref{heat-1d} is well-established and known from classical methods, see for example \cite{hsiao1979, Yong2001}. The solution of \eqref{eq:first} can be expanded as
\begin{equation} \label{eq:1heatexpansion}
u(t,x)=\tilde{u}(t,x)+\mathcal{I}(t/\epsilon^2,x)+O(\epsilon^2)
\end{equation}
where $\mathcal{I}(t/\epsilon^2, x)\sim \epsilon k \exp{-t/(\epsilon^2k)}$ is the initial layer corrector, whose presence is due to the non-equilibrium initial data $v(0,x)\not= -k\epsilon \frac{\partial u}{\partial x}$. To analyse $\mathcal{I}$, use the stretching variable $\tau=t/\epsilon^2$, and let $v(t,x) = \mathcal{I}_v(\tau, x)$, then  \eqref{eq:second} gives
\begin{equation}
  \frac{\partial \mathcal{I}_v}{\partial \tau}=
  -\frac{1}{k} \mathcal{I}_v + O(\epsilon).
\end{equation}
  By ignoring the $O(\epsilon)$ term, One has an ODE with solution $\mathcal{I}_v=v(0,x)e^{-\tau/k}$. Letting $u(t,x)=\mathcal{I}(\tau, x)$. Then \eqref{eq:first} gives
\begin{equation}\label{I-equation}
\frac{\partial \mathcal{I}}{\partial \tau}
=-\epsilon \frac{\partial \mathcal{I}_v}{\partial x }
= -\epsilon e^{-\tau/k}\frac{\partial v(0,x)}{\partial x}.
\end{equation}
Thus 
\[
\mathcal{I} \sim \epsilon k e^{-\tau/k} = \epsilon k e^{-t/(\epsilon^2k)},
\]   
namely $\mathcal{I}$ decays exponentially to zero as a function of $t/(\epsilon^2k)$.  Thus for $t>>O(\epsilon^2\ln (1/\epsilon))$, $\mathcal{I}=O(\epsilon^2)$ which can be absorbed into the error term of $O(\epsilon^2)$. 

Furthermore, the prefactor in the error of $O(\epsilon^2)$ is proportational to $\partial_{xx}u$. To see this one can use the classical  Chapman-Enskog expansion on \eqref{eq:first}-\eqref{eq:second} \cite{CLL, JinXin}:
\begin{equation}
v=-\epsilon k\frac{\partial u}{\partial x}- \epsilon^2 k\frac{\partial v}{\partial t}
=-\epsilon k\frac{\partial u}{\partial x}+\epsilon^3 k^2\frac{\partial^2 u}{\partial x \partial t} +O(\epsilon^4)
=-\epsilon k\frac{\partial u}{\partial x}-\epsilon^2 k^2 \frac{\partial^2 v}{\partial x^2}+O(\epsilon^4)=-\epsilon k\frac{\partial u}{\partial x}+\epsilon^3 k^3 \frac{\partial^3 u}{\partial x^3}+O(\epsilon^4),
\end{equation}
so the leading error for approximating $v=-\epsilon \frac{\partial u}{\partial x}$ is proportional to $k^3 \frac{\partial^3 u}{\partial x^3}$ where $u$ solves the heat equation $\frac{\partial u}{\partial t}=k\frac{\partial^3 u}{\partial x^3}$. Standard PDE analysis shows that, in a domain of unit diameter, 
$\|\frac{\partial^3 u}{\partial x^3}\|\sim e^{-2kt}$ \cite{Evans}. 
    
\end{proof}

The above lemma reveals that the approximation  \eqref{eq:vapprox} is valid for {\it any initial data} $v(0,x)$, beyond an initial layer of duration $O(\epsilon^2 \log (1/\epsilon))$. 
Thus we do not require the precise preparation for the initial state of the flux $v(0,x)$. In fact, even when we begin with $v(0,x)=0$, beyond the initial layer   $v(t,x)$  rapidly evolves to the condition in Eq.~\eqref{eq:vapprox}.  Then by solving the pair of transport equations Eq.~\eqref{eq:first} and ~\eqref{eq:second} with this simple initial condition, we can approximate the solution of the heat equation in Eq.~\eqref{eq:1heat} to precision $\epsilon^2$. \\

It is now our aim to simulate the system in Eqs.~\eqref{eq:first} and ~\eqref{eq:second} with an analog quantum simulator. We work with the system $w(t,x)=(u(t,x), v(t,x))^T$ and define the following infinite-dimensional vectors 
\begin{align}
    & \vect{u}(t)=\int u(t,x)|x\rangle dx, \qquad \vect{v}(t)=\int v(t,x)|x\rangle dx  \nonumber \\
    & \vect{w}(t)=\begin{pmatrix}
        \vect{u}(t) \\
        \vect{v}(t)
    \end{pmatrix}=|0\rangle \otimes \vect{u}(t)+|1\rangle \otimes \vect{v}(t)
\end{align}
and their corresponding quantum states
\begin{align} \label{eq:1heatquantum}
    & |u(t)\rangle=\frac{\vect{u}(t)}{\|\vect{u}(t)\|}, \qquad |v(t)\rangle=\frac{\vect{v}(t)}{\|\vect{v}(t)\|}, \qquad |w(t)\rangle=\frac{\vect{w}(t)}{\|\vect{w}(t)\|}=|0\rangle \otimes |u(t)\rangle \frac{\|\vect{u}(t)\|}{\|\vect{w}(t)\|}+|1\rangle \otimes |v(t)\rangle \frac{\|\vect{v}(t)\|}{\|\vect{w}(t)\|}, \nonumber \\
    & \|\vect{w}(t)\|^2=\|\vect{u}(t)\|^2+\|\vect{v}(t)\|^2.
\end{align}
Our first step is to simulate $|w(t)\rangle$. From Lemma~\ref{lem:1heaterror}, without losing generality, we can begin with the initial conditions $u(0,x)$, $v(0,x)=0$, which requires the initial state preparation
\begin{align}
    |w(0)\rangle=|0\rangle \otimes |u(t)\rangle. 
\end{align}
Using the rule $\partial/\partial x \rightarrow i \hat{p}_x$, then we can rewrite Eqs.~\eqref{eq:first} and ~\eqref{eq:second} as
\begin{align}
     \frac{d\vect{w}(t)}{d t}=\left(-i\frac{1}{\epsilon}\sigma_x \otimes \hat{p}_x-\frac{1}{k \epsilon^2}|1\rangle \langle 1|\otimes \mathbf{1}_x\right) \vect{w}(t)=-i\vect{A}\vect{w}(t)
\end{align}
which has the same form as Eq.~\eqref{eq:wequationfull} with 
\begin{align}
   &  \vect{A}=\vect{A}_1-i\vect{A}_2=\frac{1}{\epsilon} \sigma_x \otimes \hat{p}_x-\frac{i}{k \epsilon^2}|1\rangle \langle 1|\otimes \mathbf{1}_x, \nonumber \\
   & \vect{A}_1=\frac{1}{\epsilon} \sigma_x \otimes \hat{p}_x=\vect{A}_1^{\dagger}, \qquad \vect{A}_2=\frac{1}{k \epsilon^2}|1\rangle \langle 1| \otimes \mathbf{1}_x=\vect{A}_2^{\dagger}.
\end{align}
Since $\vect{A}$ also consists of the anti-Hermitian term $i\vect{A}_2$, we require Schr\"odingerisation. This means we include an ancilla qumode with the initial state $|\Xi\rangle$ and act on the total initial state $|w(0)\rangle \otimes |\Xi\rangle$ with the unitary generated by the Hamiltonian 
\begin{align} \label{eq:1heatham}
   &  \vect{H}=\vect{A}_2 \otimes \hat{\eta}+\vect{A}_1 \otimes \mathbf{1}_{\eta}=\frac{1}{k \epsilon^2}|1\rangle \langle 1| \otimes \mathbf{1}_x \otimes \hat{\eta}+\frac{1}{\epsilon} \sigma_x \otimes \hat{p}_x \otimes \mathbf{1}_{\eta} \nonumber \\
   &=\frac{1}{2k\epsilon^2} \mathbf{1}_2 \otimes \mathbf{1}_x \otimes \hat{\eta}-\frac{1}{2k\epsilon^2} \sigma_z \otimes \mathbf{1}_x \otimes \hat{\eta}+\frac{1}{\epsilon} \sigma_x \otimes \hat{p}_x \otimes \mathbf{1}_{\eta}. 
\end{align}
In this case, the system consists of two qumodes and one qubit degree of freedom. Note that $\hat{p}_x$ and $\hat{\eta}$ act on separate qumodes, so there are only two-party interactions between one qumode and the qubit in each interaction term. The interaction terms are of the form of the electric dipole moment interaction $\sigma_x \otimes \hat{p}$ and magnetic dipole moment interaction $\sigma_z \otimes \hat{\eta}$, which can appear for instance in a Jaynes-Cummings-like model. In the case where the interaction strengths in Eq.~\eqref{eq:1heatham} are too large, we can always rescale the Hamiltonian, and place the factors of $1/\epsilon^2$ into time, where the time of the simulation can for instance be rescaled $t \rightarrow t/\epsilon^2$. Alternatively, since we know from Lemma~\ref{lem:1heaterror} that $t \gg O(\epsilon^2 \ln(1/\epsilon))$, then in the unitary evolution $\exp(-i\vect{H}t)=\exp(-i \vect{H}_{eff})$, the effective Hamiltonian $\vect{H}_{eff}$ does not require strongly interacting terms, and the interacting strength is either $O(1)$ or $O(\epsilon)$. \\

After the evolution $\exp(-i\vect{H}t)(|w(0)\rangle \otimes |\Xi\rangle)$, one can recover $|w(t)\rangle$ by measuring the ancilla mode and postselecting on $\eta>0$ (see Section~\ref{sec:background}), which can be done for example using homodyne measurements, and has success probability $(\|\vect{u}(t)\|^2+\|\vect{v}(t)\|^2)/\|\vect{u}(0)\|^2$. We can keep $|w(t)\rangle$ as the final output of the simulator, since it contains all the physical information not only about $u(t,x)$ but also its flux $v(t,x)$. \\

If we require $|u(t)\rangle$, then one only needs to measure $|w(t)\rangle$ in the qubit degree of freedom to project onto $|0\rangle \langle 0|$, which occurs with success probability $\|\vect{u}(t)\|^2/(\|\vect{u}(t)\|^2+\|\vect{v}(t)\|^2)$, as easily seen from Eq.~\eqref{eq:1heatquantum}. Thus the total success probability from meausuring the ancilla qumode for Schr\"odingerisation and retrieving $|u(t)\rangle$ from $|w(t)\rangle$ is $O(\|\vect{u}(t)\|^2/\|\vect{u}(0)\|^2)$. This scaling is equivalent to schemes in simulating the heat equation directly, even without using the system of hyperbolic heat equation approach.\\

We can compare the quantum state output $|u(t)\rangle$ from this protocol and protocols that directly prepare $|\tilde{u}(t)\rangle=\tilde{\vect{u}}(t)/\|\tilde{\vect{u}}(t)\|$ where $\tilde{\vect{u}}(t)=\int \tilde{u}(t,x)|x\rangle dx$ and $\tilde{u}(t,x)$ is the solution of the heat equation $\partial \tilde{u}(t,x)/\partial t=k \partial^2 \tilde{u}(t,x)/\partial x^2$ with initial condition $\tilde{u}(0,x)=u(0,x)$. These states are also $\epsilon^2$-close to each other. 
\begin{theorem} \label{thm:1heat}
    Given $\tilde{u}(0,x)=u(0,x)$, $\||\tilde{u}(t)\rangle-|u(t)\rangle\|\leq O(\epsilon^2)$. Thus for some fixed $\delta>0$ where $\||\tilde{u}(t)\rangle-|u(t)\rangle| \leq \delta$, it is sufficient to choose $\delta=O(\sqrt{\epsilon})$. 
\end{theorem}
\begin{proof}
\begin{eqnarray}
\||\tilde{u}(t)\rangle-|u(t)\rangle\| &=&\left\| \frac{\tilde{u}(t)}{\|\tilde{u}(t)\|}
- \frac{u(t)}{\|u(t)\|} \right\|
=\frac{\left\|\tilde u(t) \|u(t)\|-u(t)\|\tilde{u}(t)\|\right\|}{\|\tilde{u}(t)\| \|u(t)\|}\\
&=&\frac{\|(\tilde{u}(t)-u(t))\|u(t)\|+u(t)(\|u(t)\|-\|\tilde {u}(t)\|)\|}{\|\tilde{u}(t)\| \|u(t)\|}\\
&\le &
 \frac{\|(\tilde {u}(t)-u(t)\| }{\|\tilde{u}(t)\|}+\frac{|\|\tilde {u}(t)\|-\|u(t)\||}{ \|\tilde{u}(t)\|}  \le 2 \frac{\|\tilde {u}(t)-u(t)\|}{ \|\tilde{u}(t)\|} 
 \lesssim O\left( \frac{e^{-2kt}\epsilon^2}{e^{-2kt}}\right)=O(\epsilon^2),
\end{eqnarray}
where in the last inequality the denominator is from standard estimate on the
solution for the heat equation \eqref{heat-1d} (in a  domain of unit diameter, for example) \cite{Evans}, while the upper bound in the numerator is from Lemma \ref{lem:1heaterror}.
\end{proof}

\subsection{Multidimensional heat equation} \label{sec:dheat}

We now approximate the solution to the $d$-dimensional heat equation 
\begin{align}
    \frac{\partial \tilde{u}}{\partial t}=\sum_{j=1}^d k_j \frac{\partial^2 \tilde{u}}{\partial x_j}, \qquad \tilde{u}(0,x)=u(0,x), \qquad x=(x_1, \cdots, x_d),
\end{align}
using a multidimensional extension of the hyperbolic heat equation.  Here we allow $d$ distinct parameters $\epsilon_j$, $j=1, \cdots, d$, which can be chosen to take different values to enable greater flexibility in the physical realisation of the corresponding analog quantum simulation. The corresponding system of $d+1$ hyperbolic equations is:
\begin{align} 
& \frac{\partial u}{\partial t}=-\sum_{j=1}^d \frac{1}{\epsilon_j}\frac{\partial v_j}{\partial x_j}  \label{eq:dheatsystemfirst} \\ 
& \frac{\partial v_j}{\partial t}=-\frac{1}{\epsilon_j}\frac{\partial u}{\partial x_j}-\frac{1}{k_j\epsilon_j^2}v_j, \label{eq:dheatsystemsecond} \qquad j=1, \cdots, d. 
\end{align}
Here it is clear that in the limit $\epsilon_j \ll 1$, the following approximation holds
\begin{align}
    v_j \rightarrow -k_j \epsilon_j \frac{\partial u}{\partial x_j}, \qquad j=1, \cdots, d.
\end{align}
Inserting this into Eq.~\eqref{eq:dheatsystemfirst} gives
\begin{align}
    \frac{\partial u}{\partial t} \rightarrow \sum_{j=1}^d k_j \frac{\partial^2 u}{\partial x_j^2}.
\end{align}
\begin{theorem} \label{them:dheat}
For $d$-dimensional problems, using the initial conditions $u(0,x)$ and $v_j(0,x)=0$, the error in Lemma \ref{lem:1heaterror} becomes 
\begin{equation}
\|\tilde{u}(t,x)-u(t,x)\| \leq O(e^{-2kt}d\epsilon^2),
\end{equation}
when $t \gg \epsilon^2 \ln(1/\epsilon)$ and we denote $\epsilon = \max_j \epsilon_j$. 
Similarly in $d$-dimensions, $\| |\tilde{u}(t)\rangle-|u(t)\rangle \| \leq O(d \epsilon^2)$. Thus for some fixed $\delta>0$ where $\| |\tilde{u}(t)\rangle-|u(t)\rangle \| \leq \delta$, it is sufficient to choose $\epsilon =O(\sqrt{\delta/d})$. 
\end{theorem}
\begin{proof}
    The first inequality can be easily seen from, for instance, Theorem 4.2 in \cite{Yong2001}. In $d$ dimensions, the expansion in Eq.~\eqref{eq:1heatexpansion} still holds except there are $d$ \textit{additive} terms dependent on $\epsilon_j$ in the expansion. We can write $\epsilon =\max_j \epsilon_j$ to simplify our expressions. Since there are only additive contributions due to $d$, and the rate of decay of $\mathcal{I}$ remains the same and does not depend on $d$ (this is due to the fact that $\mathcal{I}$ solves an ODE \eqref{I-equation}), so choosing $t \gg \epsilon^2 \ln(1/\epsilon)$ is sufficient. The rest of the proof to show $\||\tilde{u}(t)\rangle-|u(t)\rangle\| \leq O(d\epsilon^2)$ runs similarly to Theorem~\ref{thm:1heat}.
\end{proof}
It it is clear from Theorem~\ref{them:dheat} that since the error only scales linearly in $d$, this method does not introduce any curse-of-dimensionality to approximate $|\tilde{u}(t)\rangle$, so long as $|u(t)\rangle$ can be prepared efficiently. The latter we show below. The proof is very similar for all the $d$-dimensional parabolic PDEs so we do not repeat it in the following examples.\\

Define the hybrid $d$-qumode and single qudit ($d$ levels) quantum state 
\begin{align}
   & \vect{u}(t)=\int u(t,x)|x\rangle dx, \qquad \vect{v}_j(t)=\int v_j(t,x)|x\rangle dx \nonumber \\
   &  |w(t)\rangle=\frac{\vect{w}(t)}{\|\vect{w}(t)\|}, \qquad \vect{w}(t)=|0\rangle \otimes \vect{u}(t)+\sum_{j=1}^d |j\rangle \otimes \vect{v}_j(t), \qquad \|\vect{w}(t)\|^2=\|\vect{u}(t)\|^2+\sum_{j=1}^d \|\vect{v}_j(t)\|^2.
\end{align}
In principle, the $d$-level qudit can be rewritten as a system of $\log d$ qubits. However, as one will see from the form of our Hamiltonian, analog simulation on the qudit system is much more natural since each interaction term is only between one qudit and one qumode. If it is reduced to a system of qubits, then one would require multi-qubit entanglement, which is more difficult to experimentally realise. \\

Using the rule $\partial/\partial x_j \rightarrow i \hat{p}_j$, one can rewrite Eqs.~\eqref{eq:dheatsystemfirst} and ~\eqref{eq:dheatsystemsecond} as  
\begin{align}
    & \frac{d \vect{w}(t)}{dt}=-i\left((|0\rangle \langle 1|+|1\rangle \langle 0|)\otimes \hat{p}_1/\epsilon_1+(|0\rangle \langle 2|+|2\rangle \langle 0|)\otimes \hat{p}_2/\epsilon_2+\cdots+(|0\rangle \langle d|+|d\rangle \langle 0|)\otimes \hat{p}_{d}/\epsilon_d \right)\vect{w}(t) \nonumber \\
    &-\left(\frac{1}{\epsilon^2_1 k_1}|1\rangle \langle 1|+\cdots \frac{1}{\epsilon^2_d k_d}|d\rangle \langle d|\right)\otimes \mathbf{1}_x \vect{w}(t)=-i\vect{A} \vect{w}(t), \qquad \vect{A}=\vect{A}_1-i\vect{A}_2 \nonumber \\
& \vect{A}_1=\sum_{j=1}^d \frac{1}{\epsilon_j}(|0\rangle \langle j|+|j\rangle \langle 0|)\otimes \hat{p}_j=\vect{A}_1^{\dagger}, \qquad  \vect{A}_2=\sum_{j=1}^d \frac{1}{ \epsilon^2_j k_j}|j\rangle \langle j| \otimes \mathbf{1}_x=\vect{A}_2^{\dagger}.
\end{align}
We apply Schr\"odingerisation as before and augment our system with an ancilla qumode, and begin in the initial state $|w(0)\rangle \otimes |\Xi\rangle$. We can also choose the initial conditions $v_j(0,x)=0$ so long as $t> \epsilon_j^2\log(1/\epsilon_j)$
for all $j$. The proof is very similar to Lemma~\ref{lem:1heaterror}. It is sufficient to choose $t>\max{\epsilon_j^2 \ln(1/\epsilon_j)}$. Thus we can choose the simplest initial condition to prepare $|w(0)\rangle=|0\rangle \otimes |u(0)\rangle$.\\

Using Schr\"odingerisation, the Hamiltonian we want to simulate becomes 
\begin{align} \label{eq:dheatham}
    \vect{H}=\vect{A}_2\otimes \hat{\eta}+\vect{A}_1 \otimes \mathbf{1}_{\eta}=\sum_{j=1}^d\frac{1}{\epsilon^2_j k_j} |j\rangle \langle j| \otimes \mathbf{1}_x \otimes \hat{\eta}+\sum_{j=1}^d \frac{1}{\epsilon_j}(|0\rangle \langle j|+|j\rangle \langle 0|)\otimes \hat{p}_j \otimes \mathbf{1}_{\eta}.
\end{align}
This acts on an initial state $|0\rangle \otimes |u(t)\rangle \otimes |\Xi\rangle$ consisting of $d+1$ qumodes and a qudit with $d$ levels.  However, for each interaction term, there is only pairwise interaction between one qumode and one qudit at a time. Each interaction is again of the type $\sigma \otimes \hat{p}$. This means that we only require pairwise electric dipole moment and magnetic dipole moment-like interactions, between every two level system ${|0\rangle, |j\rangle}$ and a qumode, for every $j$. We are also allowed to tune the relative strength of each interaction with $\epsilon_j$ being different for each $j$. For instance, it is expected that for many physical systems, $1/\epsilon_j$ becomes smaller as $j$ increases. \\

Although the second term in Eq.~\eqref{eq:dheatham} is reminiscent of the Tavis-Cummings model (or the modified Dicke model), they are not equivalent models. For example, here we only want the qumode to induce shifts between the $|0\rangle$ and $|j\rangle$ qudit states for all $j$, and not between any other state pairs. \\

Similar to the one-dimensional heat equation scenario, the success probability in retrieving $|w(t)\rangle$ from $\exp(-i\vect{H}t)(|0\rangle \otimes |u(t)\rangle \otimes |\Xi\rangle)$ is $O(\|\vect{w}(t)\|^2/\|\vect{u}(0)\|^2)$. Then the total success probability to retrieve $|u(t)\rangle$ requires a projective measurement onto the $|0\rangle \langle 0|$ qudit state, so the total success probability including the Schr\"odingerisation step is again $O(\|\vect{u}(t)\|^2/\|\vect{u}(0)\|^2)$. 

\section{The Black-Scholes equation} \label{sec:blackscholes}

\subsection{The one-dimensional Black-Scholes equation}

The Black-Scholes (backward) equation \cite{black1973pricing} is a PDE that evaluates the price of a financial derivative
\begin{align} \label{eq:1blackscholes}
    \frac{\partial \tilde{u}}{\partial t}+\frac{1}{2} \sigma^2 S^2 \frac{\partial^2 \tilde{u}}{\partial S^2}+rS \frac{\partial \tilde{u}}{\partial S}-r\tilde{u}=0, \qquad \sigma, r \in \mathbb{R}, \qquad \tilde{u}(T,S)=\tilde{u}_T(x)=u_T(x)
\end{align}
where $\tilde{u}(t,S)$ is the price of the option as a function of the stock price $S$ and only the terminal condition is given. Here $r$ is the risk-free interest rate, and 
$\sigma$ is the volatility of the stock, which here we take to be constants. This one-dimensional case applies for example to single-asset options (European).\\

If one simulates Eq.~\eqref{eq:1blackscholes} directly with Schr\"odingerisation, then from \cite{2023analogPDE} one sees that this requires the realisation of highly non-Gaussian terms of the kind $\hat{x}^2 \hat{p}^2 \otimes \hat{\eta}$. Instead, we can make a transformation $S \rightarrow e^x$ and reverse time $t \rightarrow T-t$ to obtain a forward parabolic equation
\begin{align} \label{eq:bs2}
    \frac{\partial \tilde{u}}{\partial t}=\left(r-\frac{\sigma^2}{2}\right)\frac{\partial \tilde{u}}{\partial x}+\frac{\sigma^2}{2}\frac{\partial^2 \tilde{u}}{\partial x^2}-r\tilde{u}.
\end{align}
Note that since time is reversed, the terminal condition in Eq.~\eqref{eq:1blackscholes} becomes an initial condition. 
This is essentially the heat equation, now with a convection term, a linear potential term and an initial condition. Thus, following the same methodology as for the heat equation in Section~\ref{sec:heat}, we can then consider the following modified system of two linear PDEs
\begin{align} \label{eq:bssystem}
    &  \frac{\partial u}{\partial t}=-\frac{1}{\epsilon}\frac{\partial v}{\partial x}+\left(r-\frac{\sigma^2}{2}\right)\frac{\partial u}{\partial x}- r u \nonumber \\
    &  \frac{\partial v}{\partial t}=-\frac{1}{\epsilon}\frac{\partial u}{\partial x}-\frac{2}{\sigma^2 \epsilon^2} v, \qquad \epsilon \ll 1,
\end{align}
which reproduces Eq.~\eqref{eq:bs2} in the small $\epsilon$ limit. Eq.~\eqref{eq:bssystem} is a hyperbolic system, whose Jacobian matrix has two real eigenvalues $\frac{1}{2}[-(r-\sigma^2/2)\pm \sqrt{(r-\sigma^2/2)^2+4/\epsilon^2}]$,  thus is stable. Similar to the heat equation scenario, the system Eq.~\eqref{eq:bssystem} can be rewritten as 
\begin{align}
&  \frac{d \vect{w}(t)}{dt}=i\begin{pmatrix}
     r-\sigma^2/2 & -1/\epsilon \\
     -1/\epsilon & 0
 \end{pmatrix}   \otimes \hat{p}_x \vect{w}(t)+\begin{pmatrix}
     -r & 0 \\
     0 & -2/(\sigma^2 \epsilon^2)
 \end{pmatrix} \otimes \mathbf{1}_x \vect{w}(t)=-i\vect{A}\vect{w}(t), \qquad \vect{A}=\vect{A}_1-i\vect{A}_2,\nonumber \\
 &\vect{A}_1=\frac{1}{\epsilon} \sigma_x \otimes \hat{p}_x+\left(\frac{\sigma^2}{2}-r\right)|0\rangle \langle 0| \otimes \hat{p}_x=\frac{1}{\epsilon} \sigma_x \otimes \hat{p}_x+\frac{1}{2}\left(\frac{\sigma^2}{2}-r\right)(\mathbf{1}_2+\sigma_z) \otimes \hat{p}_x=\vect{A}^{\dagger}_1, \nonumber \\
 & \qquad \vect{A}_2=r|0\rangle \langle 0| \otimes \mathbf{1}_x+\frac{2}{\sigma^2 \epsilon^2}|1\rangle \langle 1| \otimes \mathbf{1}_x= \left(\frac{r}{2}+\frac{1}{\sigma^2 \epsilon^2}\right)\mathbf{1}_2\otimes \mathbf{1}_x+\left(\frac{r}{2}-\frac{1}{\sigma^2 \epsilon^2}\right)\sigma_z \otimes \mathbf{1}_x=\vect{A}_2^{\dagger}.
\end{align}
This means that the Hamiltonian we require is of the form
\begin{align}
    & \vect{H}= \vect{A}_2 \otimes \hat{\eta}+\vect{A}_1 \otimes \mathbf{1}_{\eta} \nonumber \\
    &=\left(\frac{r}{2}+\frac{1}{\sigma^2 \epsilon^2}\right)\mathbf{1}_2\otimes \mathbf{1}_x \otimes \hat{\eta}+\left(\frac{r}{2}-\frac{1}{\sigma^2 \epsilon^2}\right)\sigma_z \otimes \mathbf{1}_x \otimes \hat{\eta}+\frac{1}{\epsilon} \sigma_x \otimes \hat{p}_x\otimes \mathbf{1}_{\eta}+\frac{1}{2}\left(\frac{\sigma^2}{2}-r\right)(\mathbf{1}_2+\sigma_z) \otimes \hat{p}_x \otimes \mathbf{1}_{\eta}.
\end{align}
Here we see that the most difficult terms are still either the electric or magnetic dipole moment-like interactions in a Jaynes-Cummings model, of the kind $\sigma \otimes \hat{p}$. This is still a system of two qumodes and one qubit, where each interaction term is only between one qubit and one qumode.  \\

There is also a forward version of the Black-Scholes equation, but since the results are similar, we leave aside the details. 

\subsection{Multidimensional Black-Scholes equation} \label{sec:dbs}

The $d$-dimensional Black–Scholes equation is appropriate when one has multiple ($d$ ) underlying assets, instead of only one. It has the form
\begin{align}
    \frac{\partial \tilde{u}}{\partial t}+\frac{1}{2}\sum_{j=1}^d \sigma_j^2 S_j^2 \frac{\partial^2 \tilde{u}}{\partial S_j^2}+\sum_{j=1}^{d-1} \sum_{k=j+1}^d \kappa_{j,k}\sigma_j \sigma_k S_j S_k \frac{\partial^2 \tilde{u}}{\partial S_j \partial S_k}+\sum_{j=1}^d \mu_j S_j \frac{\partial \tilde{u}}{\partial S_j}-r\tilde{u}=0, \qquad \tilde{u}(T,S)=\tilde{u}_T(S)=u_T(S),
\end{align}
where we take $\sigma_j, \kappa_{j,k}, \mu_j$ to be constants. It can be shown that the multidmensional Black-Scholes equation can be reduced to a multidimensional heat equation. Using a similar transformation to the one-dimensional case $S_j(t) \rightarrow S_j(0)e^{x_j}$ and reversing time, one can instead solve for $\tilde{u}(t,x)$ obeying the $d$-dimensional heat equation with an initial condition (see Eq.(31) in \cite{guillaume2019multidimensional} for more details)
\begin{align} \label{eq:dbs}
    \frac{\partial \tilde{u}}{\partial t}-\frac{1}{2}\sum_{j=1}^d \sigma^2_j \frac{\partial^2 \tilde{u}}{\partial x_j^2}-\sum_{j=1}^{d-1}  \kappa_{j,j+1} \sigma_j \sigma_{j+1} \frac{\partial^2 \tilde{u}}{\partial x_j \partial x_{j+1}}-\sum_{j=1}^d \left(\mu_j-\frac{\sigma_j^2}{2}\right)\frac{\partial \tilde{u}}{\partial x_j}+r\tilde{u}=0, \quad x \in \mathbb{R}^d, \quad \tilde{u}(0,x)=\tilde{u}_0(x)=u_0(x).
\end{align}
Here the multidimensional equation features mixed derivatives (i.e. with $\kappa_{j,j+1} \neq 0$), which causes known difficulties in classical numerical schemes, for example \cite{craig1988alternating}. To tackle this, another transformation is possible $(x_1, \cdots, x_d) \rightarrow (z_1, \cdots, z_d)$ (see Proposition 2 in \cite{guillaume2019multidimensional} for details on the transformation) to reduce Eq.~\eqref{eq:dbs} to the standard multidimensional heat equation
\begin{align}
    \frac{\partial \tilde{u}}{\partial t}=\frac{1}{2}\sum_{j=1}^d \frac{\partial^2 \tilde{u}}{\partial z^2_j},
\end{align}
which can be tackled with the methods described in Section~\ref{sec:dheat}. However, this transformation is very involved, and we ask if there is an alternative method that allows us to approximate the solution to  Eq.~\eqref{eq:dbs} directly. We note that while the basic methodology from Section~\eqref{sec:heat} can be applied, there are many possible systems of one-dimensional PDEs that can approximate Eq.~\eqref{eq:dbs}. The aim is to choose the system of PDEs where we can most simply simulate this with near-term quantum systems. For simplicity we can rescale the coordinates $x_j \rightarrow x_j/\sigma_j$, so Eq.~\eqref{eq:dbs} becomes 
\begin{align} \label{eq:dbs2}
    \frac{\partial \tilde{u}}{\partial t}-\frac{1}{2}\sum_{j=1}^d \frac{\partial^2 \tilde{u}}{\partial x_j^2}-\sum_{j=1}^{d-1} \kappa_{j, j+1}\frac{\partial^2 \tilde{u}}{\partial x_j \partial x_{j+1}}-\sum_{j=1}^d \gamma_j \frac{\partial \tilde{u}}{\partial x_j}+r\tilde{u}=0, \qquad \gamma_j=\mu_j/\sigma_j-\sigma_j/2.
\end{align}
Here we see that it is a simple extension of the multidimensional heat equation with cross-terms and convection terms, which is tackled in Section~\ref{sec:heatcross}. In this case, $D_{jj}=1/2$, $D_{j,j+1}=\kappa_{j, j+1}$ and the convection terms and the linear force term are straightforward to include, without adding more complex terms to the Hamiltonian in Eq.~\eqref{eq:dheatcrossham}. Here, no matter how large $d$ is, each interaction term in the Hamiltonian involves an electric and magnetic dipole moment-like interaction between a single qudit and a qumode.\\

The multidimensional Black-Scholes equation is a just special case of the more general Eq.~\eqref{eq:generaldbs} in Appendix~\ref{sec:heatcross} when $D_{jj}=1/2$, $D_{j,j+1}=\kappa_{j,j+1}$ and is zero everywhere else. \\

As a simple illustration suppose we choose the following $d=2$ example
\begin{align}
    \frac{\partial \tilde{u}}{\partial t}-\frac{1}{2}\frac{\partial^2 \tilde{u}}{\partial x_1^2}-\frac{1}{2}\frac{\partial^2 \tilde{u}}{\partial x_2^2}-\frac{\partial^2 \tilde{u}}{\partial x_1 \partial x_2}-\gamma_1\frac{\partial \tilde{u}}{\partial x_1}-\gamma_2\frac{\partial \tilde{u}}{\partial x_2}+ru=0
\end{align}
and the following system (among many others) can be used to approximate $\tilde{u}$
\begin{align}
    &\frac{\partial u}{\partial t}=-\frac{1}{2\epsilon}\left(\frac{\partial v_1}{\partial x_1}+\frac{\partial v_2}{\partial x_1}+\frac{\partial v_1}{\partial x_2}+\frac{\partial v_2}{\partial x_2}\right)-\gamma_1\epsilon^2v_1-\gamma_2 \epsilon^2 v_2-ru \nonumber \\
    & \frac{\partial v_1}{\partial t}=-\frac{1}{2 \epsilon}\left(\frac{\partial u}{\partial x_1}+\frac{\partial u}{\partial x_2}\right)-\frac{1}{\epsilon^2}v_1 \nonumber \\
    & \frac{\partial v_2}{\partial t}=-\frac{1}{2 \epsilon}\left(\frac{\partial u}{\partial x_1}+\frac{\partial u}{\partial x_2}\right)-\frac{1}{\epsilon^2}v_2.
\end{align}
The corresponding Hamiltonian required in analog quantum simulation is then 
\begin{align}
     \vect{H} &=\frac{1}{2\epsilon}(|0\rangle \langle 1|+|1\rangle \langle 0|+|0\rangle \langle 2|+|2\rangle \langle 0|)\otimes \hat{p}_1\otimes \mathbf{1}_{\eta}+\frac{1}{2\epsilon}(|0\rangle \langle 1|+|1\rangle \langle 0|+|0\rangle \langle 2|+|2\rangle \langle 0|)\otimes \hat{p}_2\otimes \mathbf{1}_{\eta} \nonumber \\
    &-\left(\frac{1}{\epsilon^2}(|1\rangle \langle 1|+|2\rangle \langle 2|)+r|0\rangle \langle 0|+\gamma_1 \epsilon^2 (|0\rangle \langle 1|+|1\rangle \langle 0|)+\gamma_2 \epsilon^2 (|0\rangle \langle 2|+|2\rangle \langle 0|)\right)\otimes \mathbf{1}_x \otimes \hat{\eta} \nonumber \\
    &+\epsilon^2(i/2)(\gamma_1(|0\rangle \langle 1|-|1\rangle \langle 0|)+\gamma_2(|0\rangle \langle 2|+|2\rangle \langle 0|) \otimes \mathbf{1}_x \otimes \mathbf{1}_{\eta}. 
\end{align}

\section{The Fokker-Planck equation} \label{sec:fokkerplanck}

In the Fokker-Planck equation below, $\mu_j$ are the components of the drift vector and $D_{j}$ are the components of the diagonal diffusion matrix. For simplicity, keep $\mu_j$, $D_j$ as constants and we assume isotropy for diffusion 
\begin{align}
\frac{\partial u}{\partial t}+\sum_{j=1}^d \mu_j\frac{\partial}{\partial x_j} u-\sum_{j=1}^d D_j\frac{\partial^2}{\partial x_j^2} u
=0, \qquad \mu_j, D_{j} \in \mathbb{R}.
\end{align}
The solution can be approximated using the following system with $\epsilon_j \ll 1$
\begin{align} \label{eq:dfpsystem}
& \frac{\partial u}{\partial t}=-\sum_{j=1}^d \frac{1}{\epsilon_j} \frac{\partial v_j}{\partial x_j}-\sum_{j=1}^d \mu_j \frac{\partial u}{\partial x_j} \nonumber \\
& \frac{\partial v_j}{\partial t}=-\frac{1}{\epsilon_j}\frac{\partial u}{\partial x_j}-\frac{1}{D_j\epsilon_j^2}v_j, \qquad j=1, \cdots, d. 
\end{align}
It is simple to check that the system is hyperbolic -- hence stable -- since its Jacobian matrices in all directions are symmetric and has real and distinct eigenvalues. Defining $\vect{w}(t)=|0\rangle \otimes \vect{u}(t)+\sum_{j=1}^d |j\rangle \otimes \vect{v}_j(t)$ as before, we can rewrite Eq.~\eqref{eq:dfpsystem} as
\begin{align}
    & \frac{d \vect{w}(t)}{dt}=-i\sum_{j=1}^d \frac{1}{\epsilon_j}(|0\rangle \langle j|+|j\rangle \langle 0|)\otimes \hat{p}_j \vect{w}(t)-i\sum_{j=1}^d \frac{\mu_j}{\epsilon_j} |0\rangle \langle 0|\otimes \hat{p}_j \vect{w}(t)-\sum_{j=1}^d \frac{1}{D_j \epsilon_j^2}|j\rangle \langle j|\otimes \mathbf{1}_x=-i\vect{A} \vect{w}(t), \nonumber \\
    & \vect{H}=-i\vect{A}_1-\vect{A}_2, \qquad 
     \vect{A}_1=\sum_{j=1}^d \frac{1}{\epsilon_j}(|0\rangle \langle j|+|j\rangle \langle 0|+\mu_j|0\rangle \langle 0|)\otimes \hat{p}_x, \qquad \vect{A}_2=\sum_{j=1}^d \sum_{j=1}^d \frac{1}{D_j \epsilon_j^2}|j\rangle \langle j|\otimes \mathbf{1}_x.
\end{align}
Thus the Hamiltonian has the form
\begin{align}
    \vect{H}=\sum_{j=1}^d\sum_{j=1}^d \frac{1}{D_j \epsilon_j^2}|j\rangle \langle j|\otimes \mathbf{1}_x \otimes \hat{\eta}+\sum_{j=1}^d \frac{1}{\epsilon_j}(|0\rangle \langle j|+|j\rangle \langle 0|+\mu_j|0\rangle \langle 0|)\otimes \hat{p}_x \otimes \mathbf{1}_{\eta}.
\end{align}
This is a system of $d+1$ qumodes and one $d$-level qudit, and one can see that each interaction term is between one qudit and one qumode. \\

Generalisation to the anisotropic case requires the presence of cross-derivative terms. The analysis of these cases can be found in Appendix~\ref{sec:heatcross}. 

\section{Summary and discussion} \label{sec:summary}

We presented a simple protocol appropriate for analog quantum simulation that simulates quantum states embedding solutions of parabolic PDEs, like the heat, Black-Scholes and Fokker-Planck equations. This protocol allows the preparation not only of the quantum states embedding the solutions, but also the fluxes of the solutions. The interaction terms of the Hamiltonian required in analog quantum simulation are summarised in Table~\ref{tab:label}. \\

\begin{figure}[h] 
\includegraphics[width=12cm]{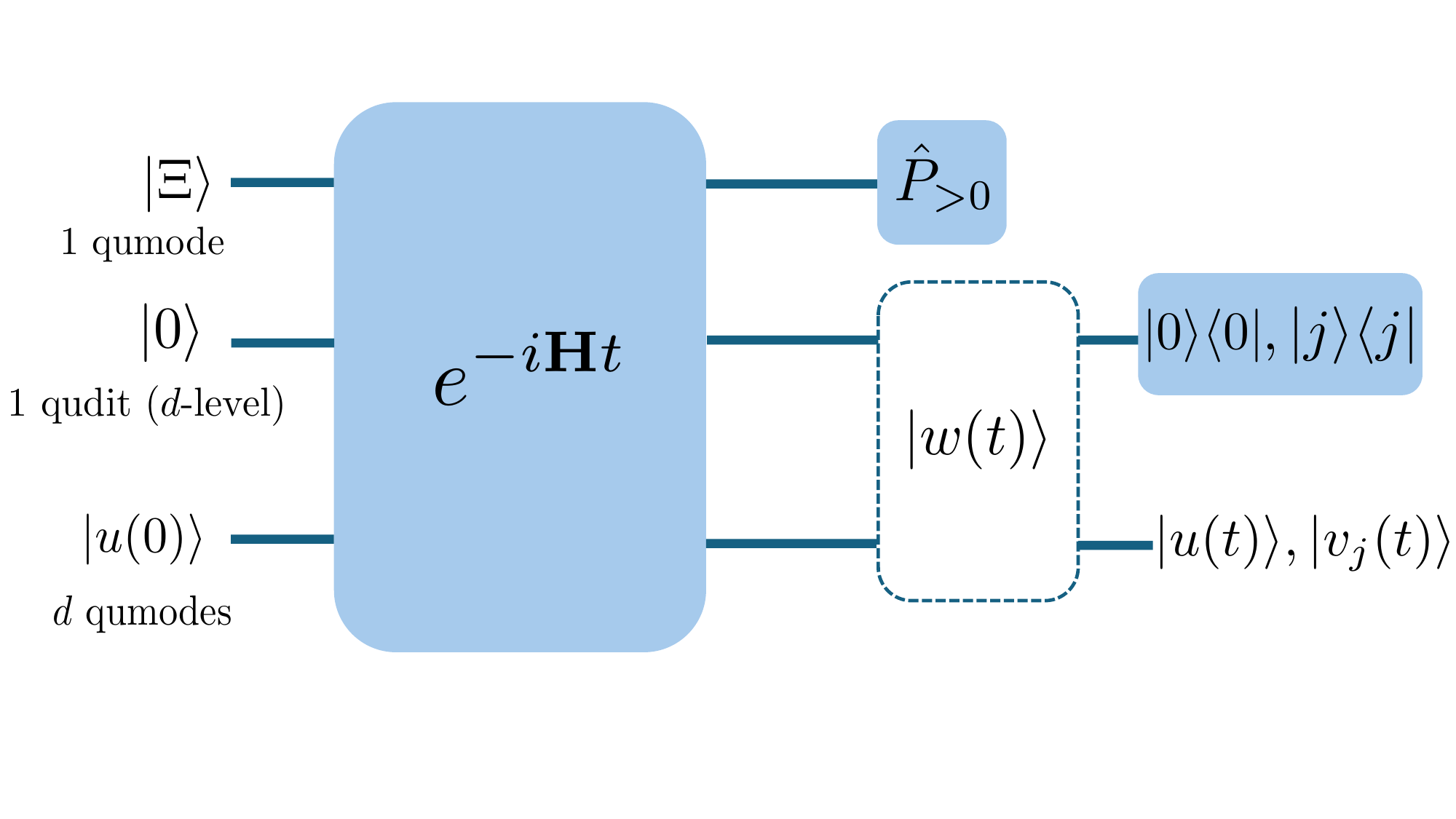} 
\caption{Protocol for preparing quantum states that embed approximate solutions $\vect{u}(t)$ of $d$-dimensional linear PDEs like the heat equation, Black-Scholes and Fokker-Planck equations. The corresponding Hamiltonian $\vect{H}$ for the different PDEs are derived in the main text, whose main interaction terms are summarised in Table~\ref{tab:label}. $|\Xi\rangle$ is the ancilla qumode requried by Schr\"odingerisation, which can be approximated by a Gaussian state whose form is fixed and not dependent on the size of the problem. After applying the unitary operator $\exp(-i\vect{H}t)$ onto initial state, the projective operator $\hat{P}_{>0}$ requires only projection onto $\eta>0$. The resulting state is $|w(t)\rangle$. To obtain $|u(t)\rangle$ one takes $|w(t)\rangle$ and makes a subsequent projective measurement onto the $|0\rangle$ qudit state. The quantum states embedding the fluxes of $u(t)$ are approximated by $|v_j(t)\rangle$, $j=1, \cdots, d$, which can be respectively obtained instead by taking $|w(t)\rangle$ and making a projective measurements onto the qudit state $|j\rangle$.}
\end{figure} 

\begin{table}[h]
    \begin{center}
    \begin{adjustwidth}{-1cm}{}
    \begin{tabular}{|c|c|c|}
        \hline
        Equation & Size of system & Hamiltonian $\vect{H}$ interaction terms \\ \hline
        1D heat & 1 qubit and 2 qumodes & $\mathbf{1}_2 \otimes \mathbf{1}_x \otimes \hat{\eta}, \quad \sigma_z \otimes \mathbf{1}_x \otimes \hat{\eta}, \quad \sigma_x \otimes \hat{p}_x \otimes \mathbf{1}_{\eta} $ \\
        1D Black-Scholes & 1 qubit and 2 qumodes & $\mathbf{1}_2 \otimes \mathbf{1}_x\otimes \hat{\eta}, \quad  \sigma_z \otimes \mathbf{1}_x \otimes \hat{\eta}, \quad \sigma_x \otimes \hat{p}_x \otimes \mathbf{1}_{\eta}$ \\
         &  & $\sigma_z \otimes \hat{p}_x \otimes \mathbf{1}_{\eta}, \quad \mathbf{1}_2 \otimes \hat{p}_x \otimes \mathbf{1}_{\eta}$ \\
        1D Fokker-Planck & 1 qubit and 2 qumodes & Same interaction terms as 1D Black-Scholes, different coefficients \\ \hline
        2D heat & 1 qudit (3 levels, qutrit) and 3 qumodes & $|j\rangle \langle j|\otimes \mathbf{1}_x \otimes \hat{\eta}$, \quad $(|0\rangle \langle j|+|j\rangle \langle 0|)\otimes \hat{p}_j \otimes \mathbf{1}_{\eta}, \qquad j=1,2 $ \\
        2D Fokker-Planck & 1 qutrit and 3 qumodes & $|j\rangle \langle j|\otimes \mathbf{1}_x \otimes \hat{\eta}, \quad (|0\rangle \langle j|+|j\rangle \langle 0|)\otimes \hat{p}_x \otimes \mathbf{1}_{\eta}, |0\rangle \langle 0|\otimes \hat{p}_x \otimes \mathbf{1}_{\eta}, \quad j=1,2$ \\ \hline
         3D heat & 1 qudit (4 levels, ququart) and 4 qumodes & $|j\rangle \langle j|\otimes \mathbf{1}_x \otimes \hat{\eta}$, \quad $(|0\rangle \langle j|+|j\rangle \langle 0|)\otimes \hat{p}_j \otimes \mathbf{1}_{\eta}, \qquad j=1,2,3$ \\
        3D Fokker-Planck & 1 ququart and 4 qumodes & $|j\rangle \langle j|\otimes \mathbf{1}_x \otimes \hat{\eta}, \quad (|0\rangle \langle j|+|j\rangle \langle 0|)\otimes \hat{p}_x \otimes \mathbf{1}_{\eta}, |0\rangle \langle 0|\otimes \hat{p}_x \otimes \mathbf{1}_{\eta}, \, j=1,2,3$  \\ \hline
        $d$-dim heat & 1 qudit d levels and d+1 qumodes & $|j\rangle \langle j|\otimes \mathbf{1}_x \otimes \hat{\eta}$, \quad $(|0\rangle \langle j|+|j\rangle \langle 0|)\otimes \hat{p}_j \otimes \mathbf{1}_{\eta}, \qquad j=1, \cdots, d$ \\  
        $d$-dim Fokker-Planck & 1 qudit d levels and d+1 qumodes & $|j\rangle \langle j|\otimes \mathbf{1}_x \otimes \hat{\eta}, \,(|0\rangle \langle j|+|j\rangle \langle 0|)\otimes \hat{p}_x \otimes \mathbf{1}_{\eta}, |0\rangle \langle 0|\otimes \hat{p}_x \otimes \mathbf{1}_{\eta}, \, j=1, \cdots, d$ \\ \hline
    \end{tabular}
    \end{adjustwidth}
    \end{center}
    \caption{A summary table of the corresponding Hamiltonians required for the analog quantum simulation of various PDEs in the $1D$, $2D$, $3D$ and $d$-dimensional scenarios. For the values of the coefficients, see the body of the paper. See Section~\ref{sec:dbs} for the Hamiltonian for the higher dimensional Black-Scholes equation.}
    \label{tab:label}
\end{table}  

For $d$-dimensional problems, we see that the size of the required quantum system scales only linearly with $d$, as does the approximation error and time. This means that if the corresponding Hamiltonian $\vect{H}$ can be realised in an analog quantum simulator, then the resource cost is linear in $d$ instead of exponential in $d$, which is more efficient than classical numerical methods which require costs exponential in $d$. \\

We note that all the key interaction terms, for any number of dimensions, are essentially of the kind $\sigma_x \otimes \hat{q}_1$ and $\sigma_z \otimes \hat{q}_2$, where $\hat{q}_1$ and $\hat{q}_2$ are quadrature operators. These correspond to the electric and magnetic dipole moment-like interactions terms as would be present in a Jaynes-Cummings-like model. One can see from Table~\ref{tab:label} that in higher dimensions, terms like $|0\rangle \langle j|+|j\rangle \langle 0|$ are essentially $\sigma_x$ except inducing state transition across the levels $|0\rangle$ and $|j\rangle$ in a qudit system. This implies that the analog quantum simulation protocol most appropriate for this setup involves qudit degrees of freedom, instead of its decomposition into qubits. Currently each interaction term in the Hamiltonian is only between one qudit and one qumode, whereas with qubits, one would instead require interacting terms with multiqubit interaction.  \\

There are various platforms that one can consider, for example cavity QED or circuit QED systems with superconductors, Rydberg atoms, acoustic modes, semiconductors, ion-traps, quantum molecular magnets and also possibly with photonic systems. \\

For proof-of-concept experiments for one-dimensional heat, Black-Scholes and Fokker-Planck equations, we require the electric and magnetic dipole moment-like interaction terms between a single qubit and a qumode. Subject to parameter tuning of the relevant terms in the Hamiltonian corresponding to the original PDE problem, these types of interactions are already in principle accessible to most of the existing platforms named above. \\

For proof-of-concept experiments for relatively low dimensional problems (two or three dimensions), we require quantum systems that can access qudits with three or four levels (qutrits and quqarts). For example, superconducting systems, ion-traps and molecular magnets  can be potential candidates, amongst others. \\

For very high $d$-dimensional problems, we require quantum systems that can realise qudits with $d$ levels, interacting with continuous modes like light or phonons. This includes for example the high number of energy levels accessible with Rydberg atoms, orbital angular momentum, frequency bins in photonic systems, and many more. \\

The purpose of this paper is only to present the basics of the idea and to deliver the message that there are multiple paths to improve the plausibility of analog quantum simulation for PDEs. We will leave detailed error analysis and plausibility analysis for different physical platforms to future work. These methods can also be straightforwardly extended to higher-order PDEs. \\

\bibliography{Ref}

\section*{Acknowledgements}
NL thanks Bill Munro for discussions and feedback on the manuscript. SJ was partially supported by the NSFC grant No. 12031013, the Shanghai Municipal Science
and Technology Major Project (2021SHZDZX0102), and the Innovation Program of Shanghai Municipal Education Commission (No. 2021-01-07-00-02-E00087). NL acknowledges funding from the Science and Technology Program of Shanghai, China (21JC1402900). Both authors acknowledge support by the NSFC grant No.12341104, the Shanghai Jiao Tong University 2030 Initiative, and the Fundamental Research Funds for the Central Universities.
\appendix 

\section{More general parabolic PDEs} \label{sec:heatcross}
We can consider more general parabolic PDEs. For example, we can consider a more general heat equation with positive definite  diffusion matrix $D=(D_{ij})_{d\times d}$:
\begin{align}
    \frac{\partial \tilde{u}}{\partial t}=\sum_{j,k=1}^d D_{jk} \frac{\partial^2 \tilde{u}}{\partial x_j \partial x_k}, \qquad \tilde{u}(0,x)=u(0,x).
\end{align}
Introducing parameters $\alpha_{ij} \in \mathbb{R}$, $1 \gg \epsilon_j>0$, we can define the following system of linear first-order PDEs
\begin{align}
    & \frac{\partial u}{\partial t}=-\sum_{i, j=1}^d \frac{\alpha_{ij}}{\epsilon_j} \frac{\partial v_i}{\partial x_j} \label{eq:heatcross1} \\
    & \frac{\partial v_i}{\partial t}=-\sum_{k=1}^d \frac{\alpha_{ik}}{\epsilon_k}\frac{\partial u}{\partial x_k}-\frac{1}{\epsilon_i^2}v_i, \label{eq:heatcross2} \qquad i=1, \cdots, d. 
\end{align}
These equations can approximate the solution $\tilde{u}(t,x)$ to precision $d^2\epsilon^2$ for $\epsilon=\max_j \epsilon_j$, if $\epsilon_j \ll 1$ and the following constraints are obeyed
\begin{align} \label{eq:dconstraint}
    D_{jk}=\sum_{i=1}^d \alpha_{ij}\alpha_{ik} \frac{\epsilon_i^2}{\epsilon_j \epsilon_k}.
\end{align}
The choices of $\alpha$ and $\epsilon$ values depend on experimental capability, and many different choices are possible that obey the same constraints. \\

Although there are many other possible systems of first-order PDEs that can also approximate the solution to $\tilde{u}(t,x)$ \cite{Natalini, Marcati, bouchut}, we will soon see the importance of choosing a system like Eqs.~\eqref{eq:heatcross1} and~\eqref{eq:heatcross2}. Defining $w(x,t)=(u(x,t), v_1(x,t), \cdots, v_d(x,t))^T$, we can rewrite Eqs.~\eqref{eq:heatcross1} and~\eqref{eq:heatcross2} as 
\begin{align} \label{eq:wheatcross}
    \frac{\partial w}{\partial t}=\sum_{j=1}^d \mathbf{M}_j \frac{\partial w}{\partial x_j}+\text{diag}(0, -1/\epsilon_1^2, \cdots, -1/\epsilon_d^2)w, \qquad w(0,x)=(u(0,x), 0, \cdots, 0).
\end{align}
Here $\mathbf{M}_j=\mathbf{M}_j^{\dagger}$ is a $(d+1)\times(d+1)$ \textit{Hermitian} matrix
\begin{align}
    \mathbf{M}_j=-\frac{1}{\epsilon_j}\sum_{k=1}^d\alpha_{kj}(|0\rangle \langle k|+|k\rangle \langle 0|).
\end{align}
The Hermiticity here is crucial. We can rewrite Eq.~\eqref{eq:wheatcross} as 
\begin{align}
    \frac{d \vect{w}(t)}{dt}=i\sum_{j=1}^d \vect{M}_j\otimes \hat{p}_j \vect{w}(t)-\sum_{j=1}^d \frac{1}{\epsilon_j^2}|j\rangle \langle j| \otimes \mathbf{1}_x \vect{w}(t)
\end{align}
and the Hermiticity of $\mathbf{M}_j$ means that the corresponding Hamiltonian used in Schr\"odingerisation is 
\begin{align} \label{eq:dheatcrossham}
& \vect{H}=\vect{A}_2 \otimes \hat{\eta}+\vect{A}_1 \otimes \mathbf{1}_{\eta}, \nonumber \\
    & \vect{A}_1=-\sum_{j=1}^d \mathbf{M}_j \otimes \hat{p}_j=\sum_{j,k=1}^d \frac{\alpha_{kj}}{\epsilon_j}(|0\rangle \langle k|+|k\rangle \langle 0|)\otimes \hat{p}_j=\vect{A}_1^{\dagger}, \qquad \vect{A}_2=\sum_{j=1}^d \frac{1}{\epsilon_j^2} |j\rangle \langle j|\otimes \mathbf{1}_x.
\end{align}
We see here that each interaction in the Hamiltonian is still an interaction between one qudit and one qumode. However, if $\mathbf{M}$ were not Hermitian, then the Hamiltonian would contain a term that is a three-body interaction between one qudit and 2 qumodes, of the form $\sigma_x \otimes \hat{p}\otimes \hat{\eta}$, where $\sigma_x$ here is a shorthand for a flip operation between two states in a qudit. For simpler experimental implementation, we therefore choose our system of first-order PDEs that give rise to a Hermitian $\mathbf{M}_j$.   \\

An even more general form of a parabolic PDE is 
\begin{align} \label{eq:generaldbs}
    \frac{\partial \tilde{u}}{\partial t}-\sum_{j,k=1}^d D_{jk} \frac{\partial^2 \tilde{u}}{\partial x_j \partial x_k}-\sum_{i=1}^d \gamma_i \frac{\partial \tilde{u}}{\partial x_i}+r\tilde{u}=0.
\end{align}
We now approximate the solution to $\tilde{u}(t,x)$ using a system of first-order PDEs. We include a new set of parameters $\delta_i$, $i=1, \cdots, d$ to Eqs.~\eqref{eq:heatcross1} and~\eqref{eq:heatcross2} in Section~\ref{sec:heatcross} to include the effect of the convection terms 
\begin{align} 
    & \frac{\partial u}{\partial t}=-\sum_{i, j=1}^d \frac{\alpha_{ij}}{\epsilon_j} \frac{\partial v_i}{\partial x_j}+\sum_{i=1}^d \frac{\delta_i}{\epsilon_i} v_i-ru, \label{eq:dbscross1} \\
    & \frac{\partial v_i}{\partial t}=-\sum_{k=1}^d \frac{\alpha_{ik}}{\epsilon_k}\frac{\partial u}{\partial x_k}-\frac{1}{\epsilon_i^2}v_i, \label{eq:dbscross2} \qquad i=1, \cdots, d. 
\end{align}
In the limit $\epsilon_j \ll 1$, Eq.~\eqref{eq:dbscross2} becomes
\begin{align}
    v_i \rightarrow -\sum_{k=1}^d \frac{\alpha_{ik} \epsilon_i^2}{\epsilon_k} \frac{\partial u}{\partial x_k}.
\end{align}
We insert this back into Eq.~\eqref{eq:dbscross1}. Comparing to Eq.~\eqref{eq:generaldbs} we see that we now have two sets of constraint equations
\begin{align} \label{eq:totalconstraint}
    & D_{jk}=\sum_{i=1}^d \alpha_{ij}\alpha_{ik} \frac{\epsilon_i^2}{\epsilon_j \epsilon_k}\nonumber \\
    & \gamma_j=-\sum_{k=1}^d \frac{\delta_j \alpha_{jk}}{\epsilon_k^2}. 
\end{align}
We note that the constraint conditions in Eqs.~\eqref{eq:dconstraint} or~\eqref{eq:totalconstraint} are not sufficient. One also needs to ensure that the first-order system \eqref{eq:heatcross1}-\eqref{eq:heatcross1}  or Eqs.~\eqref{eq:dbscross2},~\eqref{eq:dbscross2} is hyperbolic \cite{Natalini, Marcati, bouchut}, so the system is stable.  How to choose such a system with both constraints is  non-trivial and remains to be worked out, maybe  individually for different PDEs.\\

We will see that this simple model provides the basis for the analog quantum simulation of multidimensional Black-Scholes and Fokker-Planck equations with anisotropic diffusion.

%\tableofcontents 

\end{document}